\documentclass[lettersize,journal]{IEEEtran}

\usepackage{amsmath,amssymb,amsfonts}
\usepackage{array}
\usepackage{booktabs}
\usepackage{graphicx}
\usepackage{url}
\usepackage{stfloats}
\usepackage{cite}

\newtheorem{definition}{Definition}
\newtheorem{proposition}{Proposition}
\newtheorem{corollary}{Corollary}
\newtheorem{remark}{Remark}

\newcommand{\Rp}{\mathbb{R}}
\newcommand{\Cp}{\mathbb{C}}
\newcommand{\mstar}{m^\star}
\newcommand{\dmix}{\Delta_{\mathrm{mix}}}
\newcommand{\dstar}{\Delta_{\star}}
\DeclareMathOperator{\spn}{span}

\begin{document}

\title{What Selects, What Reconstructs: Repairing\\ Exemplar-Based Complex-Spectrum Separation}

\author{Maxime~Baelde%
\thanks{This work has been submitted to the IEEE for possible publication. Copyright may be transferred without notice, after which this version may no longer be accessible.}%
}

\markboth{IEEE/ACM Transactions on Audio, Speech and Language Processing}%
{Baelde: What Selects, What Reconstructs}

\maketitle

\begin{abstract}
Exemplar methods separate a mixture by picking one learned spectrum per source and deforming it until it explains the observation, making one deformation class both reconstructor and selector. We show that the second role is empty as soon as the class can interpolate: the rule then ranks candidates on its regulariser, a choice made before the data, and the estimates sum back to the mixture whichever candidate wins. The condition is a parameter count, so the diagnosis runs before any experiment. On free per-bin deformation of complex spectra it explains the observed pathologies at once: a criterion that ranks candidates by their loudness, and half an output that is a mask on the mixture rather than an exemplar. The same theorem prescribes the repair, a selection class poorer than the reconstruction class: one complex gain and one pure delay rank the candidates, and a local combination of the best-aligned atoms, fitted jointly in closed form, rebuilds them. On MUSDB18 against the exact ceiling of the masking class, the distance between the criterion and an oracle inside its own candidate pool falls under the rigid selector from 6.2--7.7 to 0.5--2.8~dB, though only 0.9--1.2~dB of that reaches the output, and the per-frame latency of the deployed rule by a factor of 47 to 806. One lock remains, quantified: atoms are scored against the mixture, so the score carries a term for the other source that absorbs the capacity the reconstruction class gains, leaving the output 10.0~dB under the ceiling. Ranking hypotheses by the residual of a fit free enough to interpolate ranks them on the regulariser alone.
\end{abstract}


\section{Introduction}

Separating a single-channel mixture from seconds of training material is the regime in which exemplar methods are attractive. Rather than fit a parametric model, they store frames of each source and explain an observed mixture frame by picking one stored frame per source and deforming the picks until they match the observation. The idea goes back to one-microphone separation by stored spectra \cite{roweis2000}, was developed into supervised and semi-supervised dictionary decompositions \cite{smaragdis2007supervised} and into exemplar dictionaries for noise-robust recognition \cite{gemmeke2011exemplar}, and it remains attractive for the same two reasons: nothing is learned beyond the dictionary itself, and inference is a search rather than an optimisation.

A search over a large dictionary is expensive, and most of the effort spent on this family has gone there, from locality-sensitive hashing over a manifold-preserving embedding \cite{kim2015lsh} to bitwise and boosted variants of the same idea \cite{kim2018bitwise,kim2020boosted}. That line accelerates the search. It says nothing about the criterion the search optimises, and a fast search of a criterion that cannot rank returns the same wrong answer faster. The criterion is the subject of this paper.

The criterion we study is that of a method pushing the exemplar idea onto complex spectra \cite{baelde2019thesis}, referred to below as Def-MAP. Each candidate pair of stored spectra is deformed by a transform free in every frequency bin, the deformation minimising a prior penalty is available in closed form, and the pair whose deformation incurs the smallest penalty is retained. The construction combines an expressive reconstruction with a closed-form inference, which is what a real-time separator requires. Measured, it separates poorly on both sources, and one frame of inference takes orders of magnitude longer than a frame of audio lasts. This paper identifies what that failure is a failure of, under a condition that is itself a parameter count and that holds wherever hypotheses are ranked by the residual of a fit rich enough to reach the data, greedy dictionary selection and residual-ranked template search among them: choosing a model by the residual it leaves then carries no information at all.

Two explanations of such a result demand opposite repairs. Either the dictionary does not carry the information, in which case the method needs more or better exemplars, or the rule that picks candidates is broken, in which case the material is there and the criterion has to be rewritten. Distinguishing them is the first contribution of this paper, and the answer is available in closed form before any experiment. The deformation grants itself more free real parameters than the frame has real observations, so every candidate pair reproduces the mixture exactly and no candidate is ever refuted by the data. What the criterion then ranks is the distance from the required correction to a reference transform, which is to say its own regulariser, chosen before the data was seen. A fit residual carries selection information only when the model class could have failed to fit, and this class never fails.

The diagnosis dictates the repair, which is the second contribution: the class that ranks and the class that rebuilds must not be the same, and the one that ranks must be the poorer of the two. We rank under a rigid, physically motivated class, one complex gain and one pure delay per atom, and rebuild under a local combination of the best-aligned atoms of each source fitted jointly by a single complex least squares. A time offset is a phase ramp linear in frequency, and phase models built on that identity or on the local consistency of a short-time spectrum are established \cite{leroux2010consistency,leroux2013consistent,magron2017anisotropic,magron2018bayesian}, the anisotropic Gaussian family being the closest antecedent in that it puts a signal-model phase prior inside the estimator. The contribution here is the division of labour rather than the ramp itself: the ramp selects, a richer class reconstructs, and a ramp asked to do both would inherit the very deficit this paper measures. Both stages stay non-iterative, and the reconstruction stage runs faster than the pair search it replaces, having dropped that search altogether.

Where the rebuilt class lands is itself a result: a local combination of aligned atoms with complex gains, free of non-negativity and non-iterative, is a close relative of non-negative matrix factorisation with a fixed basis \cite{leeseung2001,virtanen2007nmf,fevotte2009divergences}, so a supervised NMF baseline at comparable latency is reported rather than argued away. Its closest antecedent on the complex side is Complex NMF \cite{kameoka2009cnmf}, which likewise carries a phase per atom; what differs is what gets optimised, a basis and its phases learned there by iterative non-convex descent, against a frozen dictionary, phases fixed by ramp alignment and a closed-form fit here. The claim of novelty is on the division of labour between selection and reconstruction, not on the reconstruction class.

The third contribution is the measurement, on MUSDB18 \cite{musdb18} with time-domain source-to-distortion ratio after windowed overlap-add resynthesis. Most separators, classical or learned, return a real gain per bin \cite{wang2014targets,erdogan2015,williamson2016crm}. The exact ceiling of that class, and its distance from the ideal ratio mask, are established in the companion paper of this one \cite{irmbest,baelde2026ceiling}. Every row below therefore carries that ceiling alongside the ideal ratio mask, since a deficit read against the mask alone appears smaller than it is. Trained models of the task, with Open-Unmix \cite{stoter2019openunmix} as the reference point, sit above every rule discussed here, and are measured under the same protocol rather than quoted from another. Under that protocol the criterion's distance to an oracle inside its own candidate pool grows with the dictionary and collapses under the rigid class, with no change to the dictionary, the reconstruction class or the metric.

The fourth contribution is a lock the repair exposes and does not close. Selection ranks atoms by their agreement with the mixture, and the mixture is not the target, so the score of a candidate of one source carries a term measuring its agreement with the other, which vanishes only for orthogonal dictionaries. The rule therefore prefers atoms explaining the mixture to atoms resembling the source, the preference compounds as the reconstruction class grows, and the whole capacity that class gains is absorbed by the selection that populates it. Section~\ref{sec:level} reports where that leaves the absolute level, and Section~\ref{sec:lock} names the criterion that would close the lock, whose effect is quantified in advance.

\section{Residual-Based Model Selection}
\label{sec:selection}

This section builds the smallest object on which the question of the paper can be posed: a family of admissible deformations, a rule that ranks hypotheses by what they fail to explain, and the condition under which that rule carries no information at all.
\emph{Notation.} The analysis operates on frames of $L$ real samples, $L$ even, and a frame is a vector of $\Cp^F$ with $F = L/2 + 1$ the number of non-redundant bins of a real transform of that length. For $\mathbf{z} \in \Cp^F$, $\Re \mathbf{z}$ and $\Im \mathbf{z}$ denote its real and imaginary parts, both taken entrywise, so that $\Re \mathbf{z}$ and $\Im \mathbf{z}$ lie in $\Rp^F$ and $\mathbf{z} = \Re \mathbf{z} + \mathrm{i}\, \Im \mathbf{z}$, the imaginary unit being written $\mathrm{i}$ throughout so that $i$ stays available as a source index. Products between two vectors are entrywise unless a matrix is written. Source $i \in \{1,2\}$ owns a dictionary $\mathcal{C}_i \subset \Cp^F$ of $N$ learned frames, called atoms, and $\spn$ denotes the complex linear span. The observation is one mixture frame $\mathbf{x} \in \Cp^F$.

\begin{definition}[deformation class, attainable set]
\label{def:class}
A deformation class is a set $\mathcal{T}$ of maps from $\Cp^F$ to $\Cp^F$. At a candidate pair $(\mathbf{c}_1, \mathbf{c}_2) \in \mathcal{C}_1 \times \mathcal{C}_2$ its attainable set is
\begin{equation}
\mathcal{M}(\mathbf{c}_1, \mathbf{c}_2) = \left\{ T_1 \mathbf{c}_1 + T_2 \mathbf{c}_2 \ : \ T_1, T_2 \in \mathcal{T} \right\} \subset \Cp^F .
\label{eq:attainable}
\end{equation}
The class \emph{interpolates} $\mathbf{x}$ at that pair when $\mathbf{x} \in \mathcal{M}(\mathbf{c}_1, \mathbf{c}_2)$, and interpolates on $\mathcal{C}_1 \times \mathcal{C}_2$ when it does so at every pair.
\end{definition}

\begin{definition}[residual selection rule]
\label{def:rule}
Let $\Omega$ be a non-negative function on $\mathcal{T} \times \mathcal{T}$. The residual selection rule associated with $(\mathcal{T}, \Omega)$ scores a candidate pair by
\begin{equation}
\ell(\mathbf{c}_1, \mathbf{c}_2) = \min_{T_1, T_2 \in \mathcal{T}} \left\{ \Omega(T_1, T_2) \ : \ T_1 \mathbf{c}_1 + T_2 \mathbf{c}_2 = \mathbf{x} \right\} ,
\label{eq:rule}
\end{equation}
with the convention $\ell = +\infty$ when the constraint set is empty, retains a minimiser of $\ell$ over $\mathcal{C}_1 \times \mathcal{C}_2$, and returns the estimates $\hat{\mathbf{s}}_i = T_i^\star \mathbf{c}_i$ built from the minimising transforms at that pair.
\end{definition}

The rule is written in its constrained form because that is the form the method of Section~\ref{sec:free} takes. The penalised form, in which the squared residual $\lVert \mathbf{x} - T_1\mathbf{c}_1 - T_2\mathbf{c}_2 \rVert^2$ is added to $\lambda \Omega$ rather than driven to zero, behaves identically in the regime that matters here: when the constraint set is non-empty and $\Omega$ is a positive-definite quadratic, its optimal value is $\lambda\, \ell + O(\lambda^2)$ for $\lambda$ small enough, so the residual contributes to the ranking one order below the regulariser.

\begin{proposition}[selection degeneracy]
\label{prop:degeneracy}
Suppose $\mathcal{T}$ interpolates $\mathbf{x}$ on $\mathcal{C}_1 \times \mathcal{C}_2$ and that the minimum in \eqref{eq:rule} is attained at every pair. Then $\ell$ is finite everywhere, the retained pair is a minimiser of
\begin{equation}
(\mathbf{c}_1, \mathbf{c}_2) \longmapsto \Omega(T_1^\star, T_2^\star)
\label{eq:reduced}
\end{equation}
and of nothing else, and the estimates satisfy $\hat{\mathbf{s}}_1 + \hat{\mathbf{s}}_2 = \mathbf{x}$ whichever pair is retained.
\end{proposition}

\begin{IEEEproof}
Finiteness and \eqref{eq:reduced} are Definition~\ref{def:rule} read under the hypothesis: the constraint set is non-empty at every pair, so the score of every pair is the value of $\Omega$ at its own optimal transforms, and no other term enters the comparison. The summation identity is the constraint itself, $\hat{\mathbf{s}}_1 + \hat{\mathbf{s}}_2 = T_1^\star \mathbf{c}_1 + T_2^\star \mathbf{c}_2 = \mathbf{x}$, which holds at every feasible pair and therefore at the retained one.
\end{IEEEproof}

Three consequences are worth separating, because they fail in different ways. No candidate pair is ever refuted by the observation, since every pair explains it exactly, so the rule has no power to reject. The ranking it produces is whatever $\Omega$ rewards, a modelling choice fixed before the data was seen and carrying no measurement of it. And the returned estimates sum to the observation by construction, so any quantity fixed by that identity is insensitive to the dictionary, to the candidates and to the observation alike, which turns the degeneracy into something a corpus can confirm.

The hypothesis of Proposition~\ref{prop:degeneracy} is checkable before any experiment, by counting. Write $p$ for the number of real parameters that each transform of $\mathcal{T}$ carries, so that a pair carries $2p$, and $q$ for the number of real scalar constraints that $T_1\mathbf{c}_1 + T_2\mathbf{c}_2 = \mathbf{x}$ imposes; $p$ keeps that per-transform meaning throughout the paper. If the constraints are linear in those parameters and $2p > q$, the constraint set is a non-empty affine set of dimension at least $2p - q$ at every pair where the constraint matrix has full row rank, and the class interpolates by construction rather than by chance. Two cautions attend the count. Full row rank is a genuine hypothesis and its failure is not a technicality, as Section~\ref{sec:free} shows on the pairs that are silent in a bin. And $q$ is the number of independent real constraints, which is not twice the number of bins: the bins at zero and at the Nyquist frequency of a real signal's transform are purely real, so a frame of $F = L/2 + 1$ non-redundant bins carries $2F - 2 = L$ real observations rather than $2F$, one per sample of the analysed frame.

Two neighbouring literatures are worth demarcating at this point. Sparse decomposition also selects among representations that all interpolate, basis pursuit being the canonical instance \cite{chen1999basispursuit}, but it presents the penalty as the objective, whereas the rule of Definition~\ref{def:rule} presents the same quantity as a fit residual. Model selection for interpolating models likewise has criteria of its own \cite{hodgkinson2023iic}, built to correct a saturated likelihood, while the object here is a residual carrying no likelihood at all.

\section{Free Per-Bin Deformation on Complex Spectra}
\label{sec:free}

\subsection{The Method Under Study}

Section~\ref{sec:selection} is now instantiated on the method of \cite{baelde2019thesis} so that its pathologies come out as consequences of Proposition~\ref{prop:degeneracy} rather than as observations on one implementation. The analysis is a short-time Fourier transform with a periodic Hann window at half-window hop; the window is part of the feature the dictionary is drawn from, so its parameters are stated with the experiments. Def-MAP is the residual selection rule of Definition~\ref{def:rule} for one particular couple $(\mathcal{T}, \Omega)$, and the two components have to be written out separately because each carries its own half of the failure.

The class is free in every frequency bin, and its action is not a complex multiplication: the real and imaginary channels are scaled by two independent real vectors,
\begin{equation}
T \mathbf{c} = \Re \mathbf{T} \cdot \Re \mathbf{c} + \mathrm{i}\, \Im \mathbf{T} \cdot \Im \mathbf{c}, \qquad \Re \mathbf{T}, \Im \mathbf{T} \in \Rp^F ,
\label{eq:tclass}
\end{equation}
with the products taken bin by bin. The regulariser penalises the deformation itself,
\begin{equation}
\Omega(\mathbf{T}_1, \mathbf{T}_2) = \sum_{i=1}^{2} \left\| \Re \mathbf{T}_i - \mathbf{1} \right\|^2 + \left\| \Im \mathbf{T}_i \right\|^2 ,
\label{eq:loss}
\end{equation}
the squared distance from the transform to a reference transform $\mathbf{T}^\circ$ whose real channel is one and whose imaginary channel is zero. In the first frame, where the prior over dictionary indices is uniform, nothing but \eqref{eq:loss} separates candidates.

\begin{remark}[the reference transform]
\label{rem:convention}
Written in $(\Re, \Im)$ coordinates, $\mathbf{T}^\circ = (\mathbf{1}, \mathbf{0})$ is the prior mean of the maximum a posteriori formulation \eqref{eq:loss} is the negative log of: one on the real channel means keep the atom's real part as it was learned, zero on the imaginary channel means no prior phase information. It is not the identity. Applied to $\mathbf{c}$ through \eqref{eq:tclass} it returns $\Re \mathbf{c}$, so the unpenalised deformation annihilates the candidate's imaginary part and the criterion scores each pair by its distance to discarding that part. The asymmetry between the two channels is an artefact of writing a complex spectrum as two real vectors and then placing the same isotropic prior on both, and Corollaries~\ref{cor:mask} and \ref{cor:loudness} are its consequences.
\end{remark}

\subsection{The Criterion in Closed Form}

\begin{proposition}
\label{prop:closedform}
Write, at a fixed bin whose index we drop, $a_i = \Re c_i$, $b_i = \Im c_i$, $a_x = \Re x$ and $b_x = \Im x$, and take the pair non-degenerate at that bin, $a_1^2 + a_2^2 > 0$ and $b_1^2 + b_2^2 > 0$. The optimal deformation and the criterion it induces are
\begin{equation}
\Re T_i = 1 + \frac{a_i\, \varepsilon}{a_1^2 + a_2^2}, \quad \varepsilon = a_x - a_1 - a_2, \quad \Im T_i = \frac{b_i\, b_x}{b_1^2 + b_2^2},
\label{eq:transform}
\end{equation}
\begin{equation}
\ell(\mathbf{c}_1, \mathbf{c}_2) = \sum_f \frac{\varepsilon[f]^2}{a_1[f]^2 + a_2[f]^2} + \sum_f \frac{b_x[f]^2}{b_1[f]^2 + b_2[f]^2},
\label{eq:closedform}
\end{equation}
and the induced estimates are $\Re \hat{s}_i = a_i + a_i^2 \varepsilon/(a_1^2 + a_2^2)$ and $\Im \hat{s}_i = b_i^2 b_x/(b_1^2 + b_2^2)$.
\end{proposition}

\begin{IEEEproof}
Both channels are equality-constrained least squares in two unknowns. On the real channel, minimising $(T_1 - 1)^2 + (T_2 - 1)^2$ under $T_1 a_1 + T_2 a_2 = a_x$ gives $T_i - 1 = \lambda a_i$ with $\lambda = \varepsilon/(a_1^2 + a_2^2)$, hence the stated transform and a penalty $\lambda^2 (a_1^2 + a_2^2) = \varepsilon^2/(a_1^2 + a_2^2)$. On the imaginary channel the target is zero, so minimising $T_1^2 + T_2^2$ under $T_1 b_1 + T_2 b_2 = b_x$ gives $T_i = b_i b_x/(b_1^2 + b_2^2)$ and a penalty $b_x^2/(b_1^2 + b_2^2)$. Summing over bins and sources gives \eqref{eq:closedform}.
\end{IEEEproof}

\subsection{Degeneracy of the Criterion}

\begin{corollary}[exact interpolation]
\label{cor:interp}
At every bin where the candidate pair is non-degenerate, the two estimates sum back to the observation, $\hat{\mathbf{s}}_1 + \hat{\mathbf{s}}_2 = \mathbf{x}$, in both channels. The degenerate bins are those where $a_1^2 + a_2^2$ or $b_1^2 + b_2^2$ vanishes.
\end{corollary}

The count is the criterion of Section~\ref{sec:selection}, and it is immediate here: each transform of \eqref{eq:tclass} carries $p = 2F$ real parameters, so a pair carries $2p = 4F$ of them against $q = 2F - 2 = L$ independent real constraints, twice as many free parameters as the frame carries real observations, and the constraint set is a non-empty affine set of dimension at least $2F + 2$ wherever the constraint matrix has full row rank, so $\mathcal{T}$ interpolates for every dictionary and every observation. Proposition~\ref{prop:degeneracy} therefore applies verbatim, with the rank condition failing only on the bins where a candidate pair is silent, a case treated at the end of this subsection. What $\ell$ measures is not how well a pair explains the mixture, since every pair explains it perfectly, but how far the required correction sits from $\mathbf{T}^\circ$, normalised by the candidates' own energy. The two corollaries that follow are the specialisation of Proposition~\ref{prop:degeneracy} to this $\Omega$: each names one thing that \eqref{eq:loss} rewards in the absence of any data term.

\begin{corollary}[the imaginary channel is a mask]
\label{cor:mask}
$\Im \hat{s}_i$ depends on the candidates only through their imaginary energies $b_i^2$, so it is a ratio mask applied to $\Im x$, it always takes the sign of $\Im x$, and it discards every other feature of the dictionary's phase.
\end{corollary}

Half of the method's output is therefore a mask on the mixture, which is why the ceiling of \cite{baelde2026ceiling} applies to it.

\begin{corollary}[loudness bias]
\label{cor:loudness}
The imaginary term of \eqref{eq:closedform} has a numerator independent of the candidates, so at each bin the term decreases with the pair's imaginary energy $b_1[f]^2 + b_2[f]^2$, weighted by the imaginary energy the mixture holds there: it prefers the loudest atoms, whatever their shape, and two pairs of equal total energy separate only through where the mixture puts its own.
\end{corollary}

The real term is normalised by $a_1^2 + a_2^2$ as well, so a given absolute additivity defect is attenuated in proportion to the candidates' energy.

One case falls outside Proposition~\ref{prop:closedform} and is exactly where the rank condition of the count above fails. A pair silent at a bin makes both denominators vanish, the degenerate branch returns $\mathbf{T}^\circ$, and that bin contributes exactly zero, the global minimum, whatever the mixture holds there, whereas a pair of small but non-zero energy contributes a diverging penalty. The criterion is therefore discontinuous in the dictionary at its own optimum, and the dictionaries below are energy-filtered so that this branch is never the one being measured.

Two consequences of the closed form yield quantitative predictions, which Section~\ref{sec:results} tests. First, Corollary~\ref{cor:interp} couples the two sources through a single error signal, so the distance between the criterion and an oracle taken inside the same class must be equal on both sources, and there is no reason for that equality to survive a class that stops interpolating. Second, the identity $\hat{\mathbf{s}}_1 + \hat{\mathbf{s}}_2 = \mathbf{x}$ fixes the difference in reconstruction quality between the two sources independently of the dictionary's content, so a repaired rule that restores exact summation must restore the same constant. Both predictions are stated before the experiments and can be refuted by them.

On complexity, the rule of Definition~\ref{def:rule} under this class evaluates every pair and solves \eqref{eq:transform} over every bin of each, hence $O(N^2F)$ with an interpreted solve in the inner loop. Section~\ref{sec:repair} keeps the first factor and removes the second, then removes the pair search itself.

\section{Separating Selection From Reconstruction}
\label{sec:repair}

The rule proposed here runs in two stages, which must not share a deformation class. The first stage \emph{selects}: it ranks candidate pairs of atoms under a rigid class, one complex gain and one pure delay per atom, three real parameters, poor enough that the residual it leaves still measures fit. The second stage \emph{reconstructs}: on the atoms the first stage retained, it fits the complex gains of the $k$ best-aligned atoms of each source jointly by a single least squares, a class far richer than the one that ranked. What the paper proposes is the pair of stages. Where the tables below report the rigid class on its own, they report the selection stage asked to reconstruct as well, the control point of Section~\ref{sec:local} rather than a competing rule.
\subsection{Parameter Counts of the Two Stages}
\label{sec:principle}

Proposition~\ref{prop:degeneracy} is a statement about an interpolating class, so it also says what to do: make the class that ranks unable to interpolate. Selection needs a class small enough that the fit residual measures fit, and reconstruction a class large enough to be accurate, so a single class serving both roles is misspecified.

\begin{proposition}[non-degeneracy by parameter count]
\label{prop:count}
Let $\mathcal{T}$ be parametrised by $p$ real numbers per transform, with $T\mathbf{c}$ locally Lipschitz in those parameters, which holds in particular when it is $C^1$, and let $q$ be the number of independent real constraints imposed by $T_1\mathbf{c}_1 + T_2\mathbf{c}_2 = \mathbf{x}$. If $2p < q$ then $\mathcal{M}(\mathbf{c}_1, \mathbf{c}_2)$ is the image of $\Rp^{2p}$ under a locally Lipschitz map into an ambient space of dimension $q$, so its Hausdorff dimension is at most $2p < q$ and it has Lebesgue measure zero and empty interior in that space. The set of observations at which $\mathcal{T}$ interpolates is then negligible, the hypothesis of Proposition~\ref{prop:degeneracy} fails at almost every observation, and the residual of \eqref{eq:rule} recovers a data term.
\end{proposition}

The proposition converts the design problem into an inequality, and the two classes below sit on either side of it. The rigid class of Section~\ref{sec:rigid} carries three real parameters per candidate, so $2p = 6$ against $q = 2F - 2$, and it selects. The rich class of Section~\ref{sec:local} carries $2k$ real parameters per source and never selects anything; it only rebuilds, on candidates the rigid class has already ranked. Def-MAP's own class has $2p = 4F$ against the same $q$, hence the reverse inequality.

\subsection{Rigid Class: Gain and Pure Delay}
\label{sec:rigid}

The physical statement is that a library exemplar differs from the source's actual frame mainly by an unknown sub-frame time offset and a level difference. A pure delay is a phase ramp linear in frequency, so the selection class is
\begin{equation}
c[f] \longmapsto g \cdot c[f]\, e^{-2\mathrm{i}\pi f \tau / L}, \qquad g \in \Cp, \ \tau \in \Rp,
\label{eq:ramp}
\end{equation}
three real parameters against the $2F - 2$ constraints, which is Proposition~\ref{prop:count} satisfied with a wide margin. The delay is estimated by the peak of the cross-correlation $\mathrm{irfft}(\overline{\mathbf{c}}\,\mathbf{x})$, refined to sub-sample precision by a parabolic fit around that peak \cite{knapp1976gcc}, and capped. The cap follows from the model and is not a free hyperparameter: on a windowed frame the ramp is a circular shift, which only approximates a delay, and the approximation holds while the wrapped tail sits under the window's near-zero edges. Given the aligned candidates, the complex gains of a pair follow from a $2 \times 2$ Hermitian normal equation solved in closed form, regularised by a ridge proportional to its trace so that two nearly collinear candidates keep small finite gains instead of a blow-up that would win the argmin on numerical noise. Selection is the residual of that fit, and it now measures misfit because the class cannot interpolate. Proposition~\ref{prop:count} concerns the exact minimiser of \eqref{eq:rule} over this class, whereas the deployed estimator approximates it, the delay by a capped parabolic peak and the gains by a ridge-penalised solve; what the guarantee buys is the non-degeneracy of the class being searched, the quality of the search itself being measured below rather than bounded. Fig.~\ref{fig:ramp} shows the three stages on one atom.

\begin{figure}[!t]
\centering
\includegraphics[width=\columnwidth]{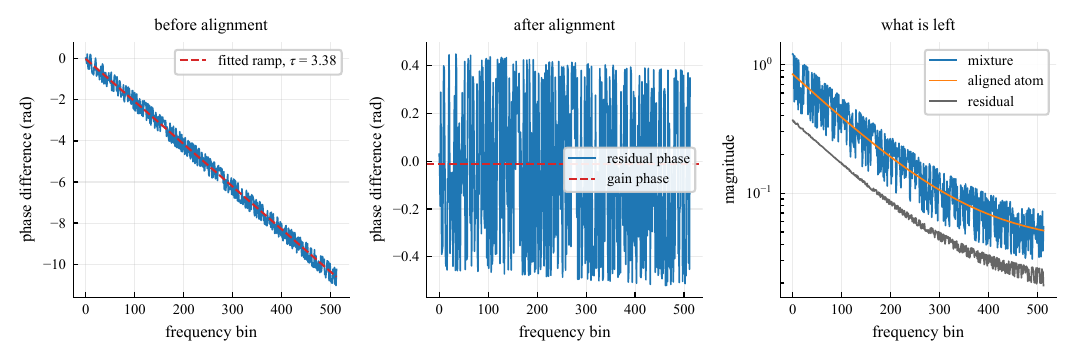}
\caption{The alignment step, on a synthetic atom of known delay and gain. Left, the phase of the atom against that of the mixture; middle, the same once the ramp \eqref{eq:ramp} is removed; right, the residual left after gain and delay are removed exactly. That residual is what the rigid class cannot explain, and what makes the criterion informative.}
\label{fig:ramp}
\end{figure}

\subsection{Rich Class: Local Combination of Aligned Atoms}
\label{sec:local}

For reconstruction we drop the pair search entirely. We align every atom of both dictionaries onto the mixture, keep the $k$ best per source, and fit all $k_1 + k_2$ complex gains jointly by a single least squares,
\begin{equation}
\hat{\mathbf{g}} = (\mathbf{A}^{*}\mathbf{A} + \rho \mathbf{I})^{-1} \mathbf{A}^{*} \mathbf{x}, \qquad \hat{\mathbf{s}}_i = \mathbf{A}_i \hat{\mathbf{g}}_i,
\label{eq:joint}
\end{equation}
with $\mathbf{A}$ the matrix of selected aligned atoms and $\rho$ the ridge of Section~\ref{sec:rigid}. At $k = 1$ the fit coincides with the pair rule of Section~\ref{sec:rigid} on the same two atoms, the joint solve reducing to the same $2\times2$ normal equations, so the sweep over $k$ starts from a control point instead of a plausible-looking curve. What differs at $k = 1$ is the selection alone, an individual matched filter here against the argmin of the joint pair residual there.

Two properties of the class matter below. The selection score is the energy-normalised matched filter $|\mathbf{c}^{*}\mathbf{x}|/\|\mathbf{c}\|$, and the normalisation is the direct countermeasure to Corollary~\ref{cor:loudness}, which is what an unnormalised score reduces to; selecting greedily by that score without deflating the mixture between picks is the first step of matching pursuit \cite{mallat1993mp} without its later ones. And the estimate of a source lives in the complex span of its selected atoms, so the ceiling of the class at a given $k$ is the projection of the true source onto a $k$-dimensional subspace. Being a projection, the capacity line of the experiments, measured with atoms aligned on the truth and chosen against it, is non-decreasing in $k$ and in dictionary size by construction, and since the subspace is chosen greedily by individual correlation rather than as the optimal $k$-subset, the measured capacity is itself a lower bound of the class ceiling.

This turns the comparison with per-bin freedom into a measurement. The local class carries $2k$ complex gains, that is $4k$ real parameters per pair, against the $4F$ the free deformation grants itself, so if a small $k$ already reaches the capacity of the free class, per-bin freedom adds no capacity that a few aligned atoms do not already provide, while leaving the criterion degenerate. Section~\ref{sec:results} measures the two capacities against each other, source by source.

\subsection{Partial Reabsorption of the Residual}
\label{sec:dial}

The joint fit leaves an unexplained residual $\mathbf{r} = \mathbf{x} - \hat{\mathbf{s}}_1 - \hat{\mathbf{s}}_2$, since the rigid class no longer interpolates. A reabsorption coefficient $\alpha \in [0, 1]$ returns a fraction of it to the estimates, split per bin by the estimates' own energies,
\begin{equation}
\hat{s}_i^{(\alpha)}[f] = \hat{s}_i[f] + \alpha\, \frac{|\hat{s}_i[f]|^2}{|\hat{s}_1[f]|^2 + |\hat{s}_2[f]|^2}\, r[f] ,
\label{eq:dial}
\end{equation}
so $\alpha = 0$ is the pure model, $\alpha = 1$ restores estimates that sum exactly to the mixture, that is, Def-MAP's own structure, and selection is always performed at $\alpha = 0$, so the model that ranks stays rigid even when the model that rebuilds does not.

One property of \eqref{eq:dial} bounds what it can be claimed to prove. The added term is a real per-bin gain applied to the residual, weighted by the model's own energy shares, hence itself of masking type: the improvement it brings is of the kind the ceiling of \cite{baelde2026ceiling} bounds, even though the composite estimate is not a mask on the mixture and that ceiling does not formally bound it. Only the $\alpha = 0$ column can therefore support a claim about the exemplar model itself, a reabsorbed residual carrying information the dictionary never had, and every such claim is read on that column below. The degeneracy and its repair do not rest on this coefficient at all.

\subsection{Selection Against the Mixture}
\label{sec:lock}

Selection ranks atoms by correlation with the mixture, and the mixture is not the target. Expanding the score of a candidate of the first source,
\begin{equation}
\langle \mathbf{c}, \mathbf{x} \rangle = \langle \mathbf{c}, \mathbf{s}_1 \rangle + \langle \mathbf{c}, \mathbf{s}_2 \rangle ,
\label{eq:crossterm}
\end{equation}
the ranking is the right criterion plus a cross-term that vanishes only if $\mathcal{C}_1$ is orthogonal to the second source, which no realistic pair of dictionaries is. So the rule prefers atoms that explain the mixture over atoms that resemble the source, and the preference compounds with $k$: as $k$ grows the selected span converges toward the best $k$-dimensional subspace for $\mathbf{x}$, not for $\mathbf{s}_1$, so the gap to capacity must widen with $k$. Section~\ref{sec:results} tests that prediction on both of its axes.

The lightest remedy, and the closest to the diagnosis, is a joint criterion over both sources: choose $(\mathbf{A}_1, \mathbf{A}_2)$ to explain $\mathbf{x}$ while penalising configurations in which the two spans overlap, that overlap being the mechanism by which explaining the mixture twice beats resembling either source once. A penalty on the principal angles between $\spn \mathbf{A}_1$ and $\spn \mathbf{A}_2$, or equivalently on $\|\mathbf{A}_1^{*}\mathbf{A}_2\|$, expresses that directly and leaves the inference structure untouched. Two approximations of it were measured on the protocol below, one discounting each atom's score by its coherence with the other source's dictionary and one selecting both sets jointly by deflation: at $k = 16$ and 300 atoms they raise the vocals from $+5.57$ to $+6.20$ and $+6.54$~dB against a capacity of $+12.99$, which closes 8 to 13 per cent of the lock and leaves it open.
\subsection{Complexity of the Two Rules}
\label{sec:cost}

The original method evaluates every pair: $N_1N_2$ candidates, each requiring the closed-form deformation over $F$ bins, hence $O(N^2F)$ with a per-pair call in the inner loop. The rigid pair rule of Section~\ref{sec:rigid} aligns the dictionary once, one grouped inverse FFT, $O(NF\log F)$ and linear in $N$. It then scores every pair through the $2\times2$ normal equations, which reduce to one Gram product $\mathbf{A}_1^{*}\mathbf{A}_2$ and are therefore $O(N^2F)$ in BLAS instead of $O(N^2)$ interpreted calls over $F$ bins. Its gain over the original is a large constant factor at the same order. Only the local combination of Section~\ref{sec:local} changes the order, because it drops the pair search entirely: alignment, then a single joint solve of size $k_1 + k_2$, that is $O(k^2F + k^3)$, the Gram of \eqref{eq:joint} dominating, with $k$ a small constant independent of the dictionary, hence $O(NF\log F)$ overall. The local combination is therefore \emph{faster} than the pair rule it replaces, and the accuracy parameter $k$ does not drive the complexity.

\section{Experimental Evaluation}
\label{sec:results}

\subsection{Protocol}
\label{sec:protocol}

The measurement layer is shared with the companion paper \cite{baelde2026ceiling}, so the numbers of the two are directly comparable.

\emph{Corpus, dictionary, splits.} MUSDB18 \cite{musdb18} for both papers of the pair, the decision being comparability rather than difficulty. The two sources are vocals and accompaniment. The dictionary is built from the training tracks with a per-track quota so that no single track dominates the pool, and test material is held out at the track level. Dictionary draws are nested, the atoms of the 50-atom cell being the first 50 of the 300-atom draw, so the trend against dictionary size is read on nested pools rather than on independent draws and cannot be an artefact of the sampling. Each test track contributes one five-second excerpt, taken at the head of the track so that no draw and no seed enters the choice, and every figure below is a paired mean over the fifty test tracks. The sensitivity of the results to the position of that excerpt is reported in Section~\ref{sec:res:diagnosis}.

\emph{Analysis and parameters.} Signals are taken at 44.1~kHz and analysed on frames of $L = 1024$ samples with hop $L/2$ and the periodic Hann window of Section~\ref{sec:free}, used again for synthesis, so $F = 513$ bins and one frame lasts 23~ms. The delay cap of Section~\ref{sec:rigid} is $\tau_{\max} = 128$ samples, 2.9~ms, and the ridge of \eqref{eq:joint} is $10^{-8}$ times the trace of the Gram. Dictionaries hold 50, 100 or 300 atoms per source, the local class is swept at $k \in \{1, 2, 4, 8, 16\}$, and $\alpha$ in \eqref{eq:dial} is swept over $[0, 1]$ in steps of $0.25$. No parameter is tuned on the test split: the two of them that could be, $k$ and $\alpha$, are reported as sweeps and never as a selected value.

\emph{Metric.} Time-domain SDR on the resynthesised waveform \cite{vincent2006bss,leroux2019sdr}. Estimates are inverted through the analysis window's own windowed overlap-add reconstruction and the first and last frame lengths are trimmed before scoring, which is not cosmetic: without that trim the round trip's edges dominate the error and even the oracles read negative. Two points depart from the campaign convention \cite{stoter2018sisec}: no distortion filter is fitted, the time-invariant filters of BSS Eval v4 absorbing the gain and phase errors measured here, and scoring is on one excerpt per track, an effect bounded in Section~\ref{sec:res:diagnosis}.

\emph{References on every row.} Five references are computed for every measurement, on the same track and the same excerpt, oracle references of this kind being the standard instrument for bounding what a class can do independently of any estimator \cite{vincent2007oracle}: the mixture taken as its own estimate, which is the floor; the ideal ratio mask; the oracle Wiener filter; the best real mask $\mstar[f] = \Re(s\overline{x})/|x|^2$, which is the true ceiling of the masking class \cite{irmbest,baelde2026ceiling}; and its clipped variant. Its distance to the two usual references is measured on this protocol rather than imported, and reported in Section~\ref{sec:level}. Every row reports two columns, the gain over the mixture $\dmix$ and the distance to $\mstar$, written $\dstar$, the second being the one that supports a claim about the model class. Table~\ref{tab:main} prints the two a level is read against, the ideal ratio mask and $\mstar$; the mixture is the origin of $\dmix$, and the oracle Wiener filter and the clipped mask are recorded with every cell without being printed.
\emph{Pairing, exclusions, reporting.} MUSDB tracks differ in difficulty by far more than the margins being measured, so every mean is paired track by track, and a rule whose reference is missing on a track is dropped instead of compared against a different one. Excerpts where one stem is silent under the other pose no separation problem and carry an infinite ceiling, one of them displacing a mean by hundreds of decibels, so they are excluded above a 60~dB threshold on the mixture's own source-to-source ratio, both sources at once so that the pairing stays symmetric. The oracle-to-criterion gap needs no separate estimator, two rules of one cell being averaged over the same tracks, so the difference of their columns is already the paired mean of the per-track gaps, as are the ceiling lines of every figure. Every cell below is read back from one record per (track, dictionary size, rule, source).

\subsection{Separating the Rule From the Dictionary}
\label{sec:res:diagnosis}

Proposition~\ref{prop:degeneracy} says the criterion of Section~\ref{sec:free} cannot rank; it does not say the dictionary is adequate, and the two explanations call for opposite repairs. We separate them by running several rules over the \emph{same} candidate pairs on held-out material: the method's own criterion, a phase-blind rule on magnitude additivity, an oracle minimising the true reconstruction error inside that pool, and the mask ceiling above. The oracle replays the criterion's own fitted transforms and oracles the choice of pair alone, so its level is that of the same reconstruction class scored under a perfect selector, not that of a deformation fitted against the truth, which under a class free in every bin would reach the truth itself. A large oracle-to-criterion gap points to the rule; a small gap at a low absolute level would point to the dictionary. One property of the setup cuts the same way: the candidate pool is flat whereas the original method indexes its library by (sound, frame), so the oracle measured here upper-bounds the oracle available to that method.

Both readings point to the rule, as Fig.~\ref{fig:diagnostic} shows. The gap is 6.23, 6.57 and 7.70~dB at 50, 100 and 300 atoms per source. The criterion's own quality peaks at a hundred atoms and ends below its own 50-atom value, $+5.84$, $+6.07$ then $+5.71$~dB above the mixture, while its oracle rises monotonically from $+12.07$ to $+13.41$~dB. The phase-blind rule on magnitude additivity, run over the same pairs, leaves a gap of 6.11, 6.34 and 6.43~dB, so discarding the phase of the dictionary altogether costs nothing the criterion has: the two rules are apart by less than the gap either of them leaves. The defect therefore lies in the rule. Nor does the diagnosis depend on where the excerpt is taken: sliding it by thirty seconds either way, on the tracks long enough to allow it, widens the gaps rather than narrowing them, under the free class from 6.2--7.7 to 7.4--9.3~dB equally on both sources, as Corollary~\ref{cor:interp} requires of any rule of that class, and under the rigid class of Section~\ref{sec:rigid} from 2.1--2.8 to 2.6--5.0~dB on vocals, the accompaniment gap of that class moving by less than a tenth of a decibel.

\begin{figure*}[!t]
\centering
\includegraphics[width=0.80\textwidth]{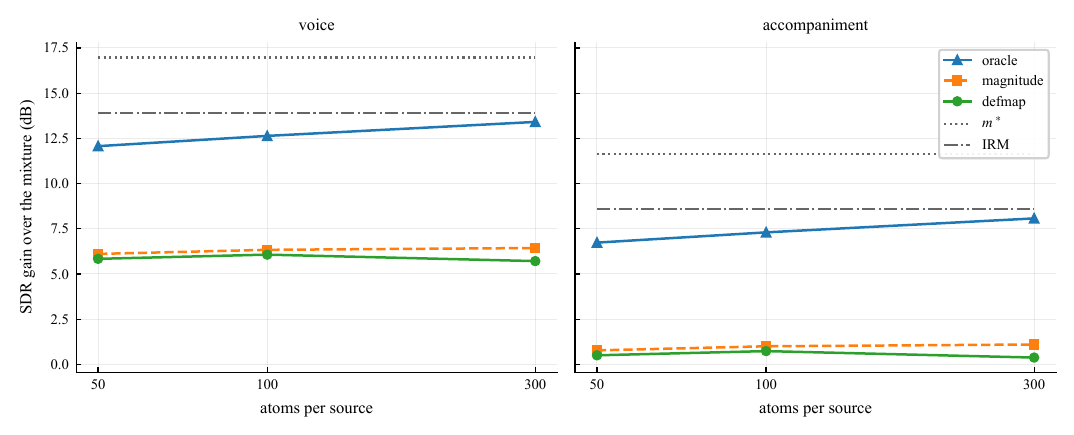}
\caption{The diagnosis. Gain over the mixture against dictionary size, both sources, for the original criterion and for an oracle picking the pair that minimises the true error in the same candidate pool. The criterion degrades as the dictionary grows while its own oracle improves, which no rule short of capacity does, and the gap is the same on both sources, as Corollary~\ref{cor:interp} predicts. Dashed lines are the derived ceilings, $\mstar$ and the ideal ratio mask, on the tracks each rule was measured on.}
\label{fig:diagnostic}
\end{figure*}

The two predictions announced at the end of Section~\ref{sec:free} are verified on the same runs. The gap is symmetric between sources to the decimal, which is the single error signal of Corollary~\ref{cor:interp} showing through; under the rigid class, which does not interpolate, it is asymmetric by a factor of three and a half to four and a half. And the reabsorption of Section~\ref{sec:dial} at $\alpha = 1$, which restores exact summation, fixes the difference in SDR between the two sources at 5.34~dB, constant to 0.01~dB across the three dictionary sizes, against 3.67 to 3.94~dB at $\alpha = 0$ and 2.0 to 2.8~dB for rules whose atoms are chosen against the truth.

\subsection{Effect of the Selection Class}

The repair of Section~\ref{sec:rigid} does not close the selection problem, it reduces it and changes its shape, and both effects are measured. Under the original criterion the oracle-to-criterion gap is the 6.23 to 7.70~dB of Section~\ref{sec:res:diagnosis}, identical on the two sources for the reason given there. Under the rigid class the same gap falls to 2.13, 2.29 and 2.80~dB on vocals and to 0.48, 0.55 and 0.80~dB on accompaniment, as Fig.~\ref{fig:constraint} shows. Two thirds to nine tenths of the loss is therefore attributable to the parametrisation of the selection class alone, with no change to the dictionary, to the reconstruction class or to the metric.

\begin{figure*}[!t]
\centering
\includegraphics[width=0.80\textwidth]{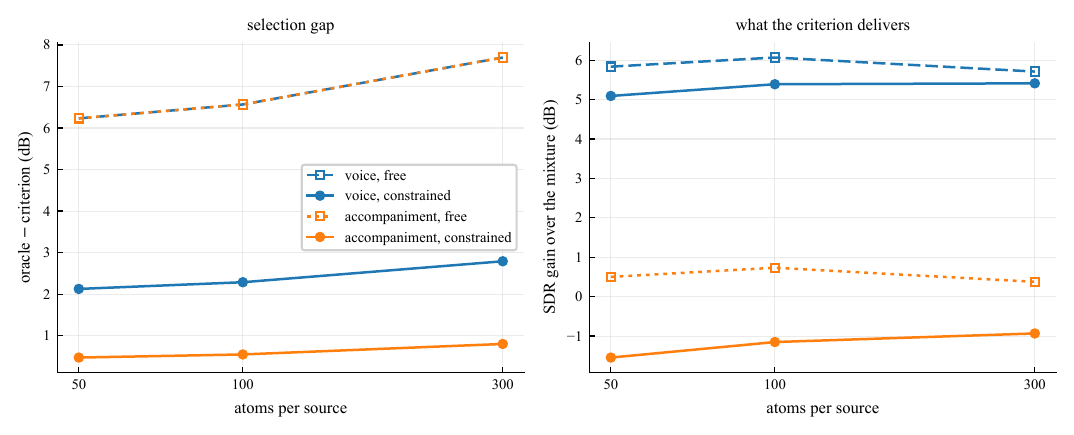}
\caption{The selection gap before and after the phase constraint. Left, the oracle-to-criterion gap against dictionary size, under the free per-bin deformation and under the rigid class \eqref{eq:ramp}, both sources; right, the quality each criterion delivers. The gap collapses and its symmetry between sources breaks, the signature Corollary~\ref{cor:interp} predicts, while the criterion's own quality now grows with dictionary size.}
\label{fig:constraint}
\end{figure*}

What survives is small in absolute terms and unambiguous in trend. The residual gap \emph{grows} with the dictionary on both sources, by 0.67~dB on vocals and 0.32~dB on accompaniment between 50 and 300 atoms, and it does so while the oracle of the same class improves over the same range, which no rule short of capacity does. This is the prediction of \eqref{eq:crossterm}: a larger dictionary offers more atoms that exploit the cross-term, and the rule takes them. The criterion therefore ranks, and what remains is that it ranks against the wrong target.

\subsection{Quality Against Capacity as $k$ Grows}

The same lock is read a second time on the reconstruction class, where it is larger and where the sweep over $k$ exposes its mechanism. Capacity, atoms chosen against the truth, and quality, atoms chosen against the mixture, diverge monotonically as $k$ grows: on vocals at 300 atoms per source the gap is 2.75, 3.31, 4.12, 5.34 and 7.42~dB for $k = 1, 2, 4, 8, 16$, and on accompaniment 0.85, 1.06, 1.53, 2.56 and 4.53~dB. It also grows with the dictionary at fixed $k$, from 2.09 to 2.75~dB at $k = 1$ and from 6.36 to 7.42~dB at $k = 16$, so the compounding predicted by \eqref{eq:crossterm} is measured on two axes at once.

Over the same sweep, capacity rises from $+8.56$ to $+12.99$~dB above the mixture while the quality actually delivered stays flat, $+5.81$, $+6.02$, $+6.11$, $+6.05$, then $+5.57$~dB, peaking at $k = 4$ and falling at $k = 16$; Fig.~\ref{fig:selection} hatches the difference. The whole of the capacity that the reconstruction class gains from a larger $k$ is lost again by the criterion that populates it. So $k$ does not control accuracy under the present selection rule. The operating point of Table~\ref{tab:main} is $k = 2$ or $4$ rather than the largest value measured. Two independent measurements, a pair rule with three parameters and a local combination with $k_1 + k_2$, widen against their own oracles along the same two axes, so the residual cause is the target of the criterion rather than the size of either class.

\begin{figure*}[!t]
\centering
\includegraphics[width=0.80\textwidth]{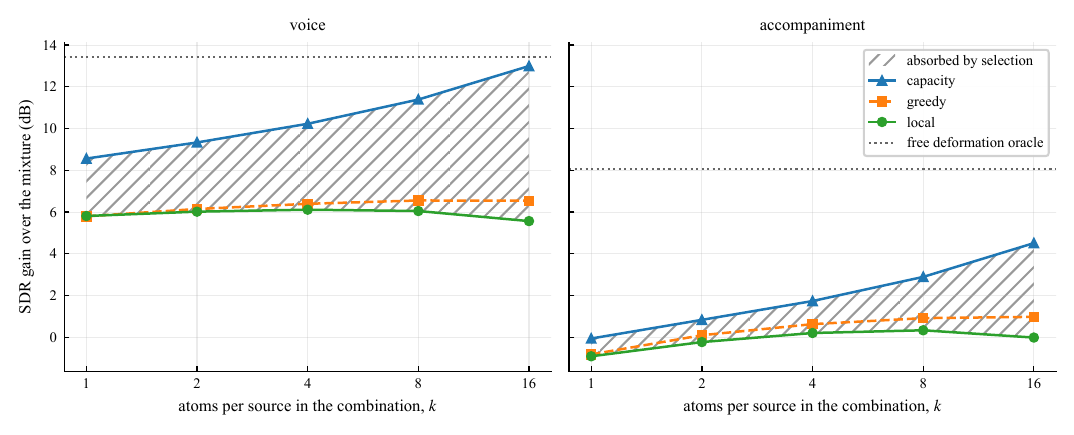}
\caption{Quality against capacity as $k$ grows, both sources, 300 atoms per source. Capacity is the projection of the true source onto the span of $k$ atoms aligned and chosen against the truth; quality is the same class populated by the deployed rule. The hatched selection lock absorbs the entire 4.4~dB that capacity gains between $k = 1$ and $k = 16$.}
\label{fig:selection}
\end{figure*}

\subsection{Absolute Level and Baselines}
\label{sec:level}

Table~\ref{tab:main} puts the repaired rule against the original method, against two baselines and against the references of Section~\ref{sec:protocol}. The absolute level is read against $\mstar$, the ceiling of the masking class, rather than against the ideal ratio mask: on this protocol the mask and the oracle Wiener filter both sit about 3~dB below $\mstar$, so a deficit phrased against the mask understates it by that margin. The best repaired rule of the table sits 10.0~dB below $\mstar$ and 7.0~dB below the ideal ratio mask on vocals at 300 atoms. Its paired gain over the original method is 0.9 to 1.2~dB across the three dictionary sizes, with a between-track deviation of 1.3 to 1.4~dB, hence a paired standard error of about $0.19$~dB over the fifty tracks and a gain standing at four to six standard errors. The gain on the quality axis is therefore modest, the 6.2 to 7.7~dB of criterion loss measured in Section~\ref{sec:free} converting only in part. The oracle taken inside the same local class gains 3.4 to 4.5~dB over the original method at $k = 4$ and 6.1 to 7.3~dB at $k = 16$, against the 0.9 to 1.2~dB the deployed rule delivers. Closing that distance is the subject of Section~\ref{sec:lock}.

The two baselines answer two different objections. Supervised NMF with a fixed basis learned on exactly the frames the dictionary is drawn from is the comparison at comparable latency: 17.9~ms converged and 3.2~ms at twenty-five updates, against 10.7~ms for the repaired rule at 100 atoms per source. The structural proximity established in Section~\ref{sec:local} makes that comparison necessary. Open-Unmix \cite{stoter2019openunmix} is the comparison against a trained model of the task, retained over the stronger learned separators published since because it can be rerun end to end under our own protocol; on the same tracks, the same excerpts and the same metric, it reaches $+11.6$~dB over the mixture on vocals, 5.4~dB below $\mstar$ and 2.3~dB below the ideal ratio mask. Being bidirectional over a whole excerpt, it has no per-frame latency to report, and its mask is applied in its own 4096-point transform while we score at 1024. Two diagnostics confirm that its output stays in the masking class, three degrees of median phase deviation and two per cent of gains above one, so its $\dstar$ reads as a distance to the ceiling of our class.

\begin{table*}[!t]
\caption{Main results on the MUSDB18 test split, fifty tracks on five-second excerpts, 300 atoms per source where a dictionary applies. No excerpt was dropped by the silence guard, no track by a missing reference. $\dmix$ is the gain over the mixture, $\dstar$ the distance to the best real mask, both paired track by track. Latency is the per-frame time of the deployed rule alone. Rules whose estimates sum to the mixture read the same $\dstar$ on both sources: their per-source SDR difference is then the constant of Section~\ref{sec:res:diagnosis}, 5.33~dB here, also the gap between the two $\mstar$ columns.}
\label{tab:main}
\centering
\begin{tabular}{l l r r r r r}
\toprule
& & \multicolumn{2}{c}{vocals} & \multicolumn{2}{c}{accompaniment} & \\
\cmidrule(lr){3-4} \cmidrule(lr){5-6}
rule & class & $\dmix$ & $\dstar$ & $\dmix$ & $\dstar$ & latency \\
\midrule
Def-MAP \cite{baelde2019thesis}      & free per bin      & $+5.71$  & $-11.26$ & $+0.38$ & $-11.26$ & 1037.2 \\
oracle pair, free deformation        & free per bin      & $+13.41$ & $-3.56$  & $+8.07$ & $-3.56$  & n/a \\
rigid pair, $k = 1$                  & gain and delay    & $+5.41$  & $-11.56$ & $-0.93$ & $-12.57$ & 12.2 \\
local combination, $k = 4$           & local, $\alpha=0$ & $+6.11$  & $-10.86$ & $+0.20$ & $-11.43$ & 10.7 \\
local combination, $k = 4$, $\alpha > 0$ & local, $\alpha=0.75$ & $+6.93$ & $-10.04$ & $+1.47$ & $-10.16$ & 10.7 \\
capacity of the local class, $k = 4$ & local, oracle     & $+10.23$ & $-6.74$  & $+1.73$ & $-9.90$  & n/a \\
\midrule
supervised NMF, Wiener form          & mask              & $+7.72$  & $-9.25$  & $+2.39$ & $-9.25$  & 17.9 \\
supervised NMF, 25 updates           & mask              & $+7.82$  & $-9.15$  & $+2.49$ & $-9.15$  & 3.2 \\
Open-Unmix \cite{stoter2019openunmix} & trained          & $+11.61$ & $-5.36$ & $+6.27$ & $-5.36$  & n/a \\
\midrule
ideal ratio mask                     & mask              & $+13.92$ & $-3.05$  & $+8.59$ & $-3.05$  & n/a \\
best real mask $\mstar$              & mask              & $+16.97$ & $0.00$   & $+11.64$ & $0.00$  & n/a \\
\bottomrule
\end{tabular}
\end{table*}

\subsection{Delivered Quality by Source}

Read on the $\alpha = 0$ column, as Section~\ref{sec:dial} requires, the exemplar model converts training material into quality far more slowly than it gains capacity. At $k = 1$ the vocals gain 0.34~dB between 50 and 300 atoms per source, from $+5.47$ to $+5.81$~dB over the mixture, while the capacity of the same class gains 1.00~dB, from $+7.56$ to $+8.56$; at $k = 16$ the progression is not even monotone, $+5.56$, $+5.84$, $+5.57$. More material therefore raises the capacity of the class without raising what is delivered, which is the lock of Section~\ref{sec:lock} seen from a third angle. The comparison of capacities announced in Section~\ref{sec:local} settles on vocals: at $k = 16$ the local class matches the free per-bin oracle, short of it by 0.15, 0.03 and 0.42~dB at the three sizes, with $2k = 32$ complex gains, 64 real parameters per pair, against the $4F = 2052$ the free class grants itself, and from a pessimistic estimate of its own ceiling. One qualification bounds that: the three deficits are not ordered by dictionary size, so the claim holds at $k = 16$ and at the sizes measured rather than asymptotically.

On accompaniment the reserve of Section~\ref{sec:dial} removes one claim outright. At $k = 1$, where the reconstruction class coincides with the pair rule, the exemplar model alone sits \emph{below} the mixture taken as its own estimate, at $-1.53$, $-1.17$ and $-0.91$~dB, and the capacity of the class stays below the mixture at all three sizes as well, at $-1.01$, $-0.53$ and $-0.06$~dB; the local class does not approach the free oracle there either, staying 2.7 to 3.6~dB under it. Raising $k$ to 4 lifts both above the mixture, $+0.20$~dB delivered against $+1.73$~dB of capacity, by margins an order of magnitude smaller than the vocals figures of Table~\ref{tab:main}. The positive gain reported on that source, $+1.35$, $+1.65$ and $+1.55$~dB at the interior optimum $\alpha = 0.75$ and $k = 1$, therefore comes from the reabsorption term, which \eqref{eq:dial} identifies as a soft mask; at $k = 4$ the model carries $+0.20$ of the $+1.47$ measured at the same $\alpha$, and the mask the rest. The median phase error separates the two mechanisms, 83 degrees for the model against 9 to 13 degrees once the residual is reabsorbed, the phase then coming from the mixture rather than from the dictionary. Hence the scope: the degeneracy and its repair hold on both sources, exemplar-based separation is claimed on vocals alone.

\subsection{Measured Latency}

Latencies depend on $(N_1, N_2, F)$ and never on the audio, so they are measured on random spectra without any corpus, timing the inference rule alone and excluding the diagnostic rules a deployed separator never evaluates. The measurement is best-of-five on unoptimised single-frame numpy, on one core of an AMD Ryzen~7 7435HS, two conventions pulling in opposite directions, so the absolute margins against the 23~ms frame duration are indicative and what Table~\ref{tab:cost} supports is the comparison between rules, all timed the same way. The Def-MAP column is extrapolated from a timed sample of pairs, 300 atoms per source meaning ninety thousand interpreted solves per frame.

\begin{table}[!t]
\caption{Per-frame latency in milliseconds, single core, $F = 513$; one frame of 1024 samples at 44.1~kHz lasts 23~ms. The last column splits the repaired rule between alignment and pair scoring.}
\label{tab:cost}
\centering
\setlength{\tabcolsep}{4pt}
\begin{tabular}{r r r r r l}
\toprule
atoms/src & Def-MAP & ramp & local $k{=}4$ & local $k{=}16$ & align + score \\
\midrule
50   & 243.4    & 5.4   & 5.2   & 5.9   & 4.3 + 1.2 \\
100  & 1037.2   & 12.2  & 10.7  & 11.6  & 9.1 + 2.2 \\
300  & 9513.4   & 65.7  & 44.7  & 45.3  & 39.0 + 26.3 \\
1000 & 103798.0 & 316.4 & 128.8 & 128.2 & 116.6 + 180.3 \\
\bottomrule
\end{tabular}
\end{table}

The repair yields a factor of 47 to 806 on the local combination, the rule the abstract and Table~\ref{tab:main} report, and 45 to 328 on the rigid pair rule, which only ranks. Def-MAP is quadratic throughout, its latency multiplying by 4.3, 9.2 and 10.9 for dictionary ratios whose squares are 4, 9 and 11.1, while the ramp rule is alignment-dominated at 50 atoms and quadratic later, so the ratio grows and then saturates. The local combination is the faster of the two repaired rules, the gap widening with the dictionary exactly because the pair search is the quadratic term, and it is flat in $k$ to within a millisecond over the whole range measured, so raising $k$ leaves the latency unchanged. The frame duration itself is met up to 100 atoms per source and exceeded by a factor of 1.9 at 300, which is the size carrying the best quality figures, so real-time feasibility holds at the smaller sizes only. The last column locates where an optimisation would act, alignment dominating pair scoring up to 300 atoms and the ordering flipping at 1000, where the quadratic term overtakes the linear one.

\subsection{Scope and Limits}

Three limits bound the above, and the first two are properties of the repaired rule rather than of the measurement. Each atom's delay is estimated against the mixture and is therefore contaminated by the other source; the measured distance to capacity being small at the operating point this has not bitten here, and the clean path is to align, fit, subtract the other source's current estimate and realign. The estimator adds a coupling of its own, the complex gain fitted after the delay rotating the correlation and displacing its peak by roughly the gain's phase divided by the atom's centre frequency; decoupling it requires refining on the analytic envelope, hence another rule. And with the transition term of the original method the model is a factorial hidden Markov model \cite{ghahramani1997fhmm,roweis2000,radfar2019gainadapted} whose exact decoding is Viterbi in $O(N^2)$ per frame, the same quadratic growth as the pair search, and the theory above does not depend on it.

What this paper does not claim follows from the same tables. Not state-of-the-art quality: the repaired rule stays 10.0~dB below the ceiling of the masking class on vocals, and a trained model of the task is 4.7~dB closer to that ceiling. Not a rule dominating its own NMF baseline: at the operating point the twenty-five-update supervised NMF is ahead on both reported axes, quality and latency, and what the exemplar route buys is the diagnosis and a model carrying its own phase, not the level. Not exemplar-based separation on accompaniment. Not a blanket real-time claim.

\section{Conclusion}

Choosing a model by the residual it leaves is empty as soon as the class that produced the residual can interpolate the observation. Proposition~\ref{prop:degeneracy} states it in that generality, its hypothesis is a parameter count, and nothing in it is specific to audio or to exemplars: any procedure ranking hypotheses by the penalty of a fit flexible enough to reach the data is ranking on its regulariser, and the ranking is whatever that regulariser happens to reward.

An exemplar method on complex spectra is the instance that motivated the statement. With $4F$ free real parameters under $2F - 2$ constraints every candidate interpolated the mixture exactly, so its criterion preferred loud atoms and scored a silent pair perfectly, its dictionary and its deformation both being sound on vocals. Splitting the class in two, as the same proposition prescribes, recovers two thirds to nine tenths of the selection gap under the rigid selector and runs 47 to 806 times faster: three real parameters select, a local combination of aligned atoms rebuilds, and the richer reconstruction runs faster than the pair search it replaces. Only 0.9 to 1.2~dB of that gap reaches the output. What survives is one identified mechanism, selection against the mixture instead of against the source, which absorbs the entire capacity the reconstruction class gains and is now quantified at 7.4~dB. Its remedy is named and left open, a criterion penalising the overlap of the two spans whose approximations recover a tenth of it; anything that closes it converts directly, the capacity being already present.

\section*{Acknowledgment}

The companion code, the measurement layer and the journals behind every number in this paper are available at \url{https://github.com/mbaelde/defmap-repair}, tag \texttt{v1.0.0}.

\end{document}